\documentclass[aps,pra,showpacs,floatfix,superscriptaddress,twocolumn]{revtex4-1}
\usepackage{graphicx}
\usepackage{dcolumn}
\usepackage{amssymb,amsmath,amsfonts,mathrsfs,amsthm}
\usepackage{dsfont}
\usepackage{bm}
\usepackage{epstopdf}
\usepackage{color}

\def\dim{\mathsf{dim}}

\def\<{\langle}
\def\>{\rangle}

\def\set#1{{\sf #1}}

\def\map#1{{\mathcal{#1}}}
\def\hil#1{{\mathscr{#1}}}

\def\Supp{\set{Supp}}

\def\Reals{\mathbb R}

\def\Naturals{\mathbb N}

\def\Tr{\operatorname{Tr}}

\DeclareMathOperator*{\SumInt}{%
\mathchoice%
  {\ooalign{$\displaystyle\sum$\cr\hidewidth$\displaystyle\int$\hidewidth\cr}}
  {\ooalign{\raisebox{.14\height}{\scalebox{.7}{$\textstyle\sum$}}\cr\hidewidth$\textstyle\int$\hidewidth\cr}}
  {\ooalign{\raisebox{.2\height}{\scalebox{.6}{$\scriptstyle\sum$}}\cr$\scriptstyle\int$\cr}}
  {\ooalign{\raisebox{.2\height}{\scalebox{.6}{$\scriptstyle\sum$}}\cr$\scriptstyle\int$\cr}}
}

\newtheorem{lemma}{Lemma}[section]

\newtheorem{theorem}{Theorem}

\begin{document}
\title{Jaynes' principle for quantum Markov processes: Generalized Gibbs - von Neumann states rule}
\author{J Novotn\'y, J Mary\v ska and I Jex}

\address{Department of Physics, Faculty of Nuclear Sciences and Physical Engineering,
Czech Technical University in Prague, B\v rehov\'a 7,
115 19 Praha 1 - Star\'e M\v{e}sto, Czech Republic \\ corresponding author's email:jaroslav.novotny@fjfi.cvut.cz
}

\date{\today}




\date{\today}

\begin{abstract}
We prove that any asymptotics of a finite-dimensional quantum Markov processes can be formulated in the form of a generalized Jaynes' principle in the discrete as well as in the continuous case. Surprisingly, we find that the open system dynamics does not require maximization of von Neumann entropy. In fact, the natural functional to be extremized is the quantum relative entropy and the resulting asymptotic states or trajectories are always of the exponential Gibbs-like form. Three versions of the principle are presented for different settings, each treating different prior knowledge: for asymptotic trajectories of fully known initial states, for asymptotic trajectories incompletely determined by known expectation values of some constants of motion and for stationary states incompletely determined by expectation values of some integrals of motion. All versions are based on the knowledge of the underlying dynamics. Hence our principle is primarily rooted in the inherent physics and it is not solely an information construct.
The found principle coincides with the MaxEnt principle in the special case of unital quantum Markov processes. We discuss how the generalized principle modifies fundamental relations of statistical physics.
\end{abstract}

\maketitle
\section{Introduction}
\label{sec:introduction}
Extremal and variational principles \cite{Lanczos,Basdevant,Kirchhoff,Gibbs,Gerjvoy,Rojo, MaxEnt,Jaynes1,Jaynes2,Jaynes3,Kapur1,Kapur2,Banavar2010} play a central role in
the formulation of the laws of Nature and pervade all branches of science ranging from evolutionary biology to analytic mechanics or quantum field theory. They represent a common conceptual frame, in many cases they have a very broad range of applicability and simplify solution of a given task. They improve our understanding of the studied field and consequently allow for possible, sometimes quite far reaching, extensions. 

A large class of such principles \cite{MaxEnt,Jaynes1,Jaynes2,Jaynes3,Kapur1,Kapur2,Banavar2010} have been found extremely useful if the complexity of the investigated problem prohibits to reveal its solution by standard techniques, e.g. by directly solving equations of motion or the solution contains too much information about the system to be used for its effective control. Such situation naturally arises in statistical physics, where the number of mutually interacting subsystems is gargantual. It does not improve in the quantum domain or for open system dynamics. An open system quantum dynamics eventually drives macroscopic system into equilibrium or at least a stationary regime. Due to high complexity of the whole process, also frequently accompanied with a lack of details of the dynamics and of the initial state, the analysis of equilibrium mechanisms and the emerging states constitutes a formidable task. In 1957 Jaynes \cite{Jaynes1,Jaynes2,Jaynes3}, to circumvent the whole issue, introduced the MaxEnt principle turning the problem into finding a constrained extreme of entropy. It affirms that among all states consistent with observed data the one with the maximal entropy best represents the actual knowledge of the system state. Thus, provided that mean values $\<A_i\>_{\rho} = \Tr\left[\rho A_i\right]$ of observables $A_i$ (typically conserved quantities) constitute the only available knowledge about the established equilibrium, its corresponding state, maximizing the von Neumann entropy $S(\rho)= - \Tr\left[\rho \ln \rho\right]$ (with Boltzmann constant $k_B=1$), reads \cite{Balian}

\begin{equation}
\label{eq_equilibrium_state}
\rho_{eq}=\frac{1}{Z} \exp\left(-\sum_i \gamma_i A_i\right).
\end{equation}
$Z= \Tr\{\exp(-\sum_i \gamma_i A_i)\}$, the normalizing partition function as a function of independent Lagrange multipliers $\gamma_i$ comprises all relevant information necessary to construct statistical mechanics characteristics of the equilibrium, i.e. mean values
$\<A_i\>_{\rho_{eq}} = - \left(\partial \ln Z /\partial \gamma_i\right)_{\gamma_{j \neq i}}$, entropy $S  = \ln Z + \sum_i \gamma_i \<A_i\>_{\rho_{eq}}$ of the equilibrium state (\ref{eq_equilibrium_state}) and the change of entropy $dS = \sum_i \gamma_i d\<A_i\>$ for a shift of the equilibrium.

The principle was applied to many problems \cite{Banavar2010,Karmeshu2012}, but also attracts significant criticism and raises relevant questions at the same time. To which dynamical systems does the Jaynes' principle actually apply and what is its explicit form? Which available information, represented by a set of observables $A_i$, must be taken into considerations? Does the von Neumann entropy play a unique role or may or should we extremize some other (entropy like) quantity? For example, the Tsallis entropy is used to build up non-extensive statistical thermodynamics \cite{Tsallis2009}. A minimum relative entropy principle is frequently used in adaptive strategies \cite{Ortega2010}. However, this requires additional {\it prior} information (reference probability distribution), which must be obtain in a different way and it is not clear, in general, how the prior distribution should be chosen. Another question is whether we can generalize the principle also to trajectories, so called the maximum calibre principle \cite{Ghosh}, which could enlarge its scope to non-equilibrium statistical physics. Finally, the main objection urges that such a principle should be consistent with the underlying dynamics and in principle we should be able to derive it from microscopic laws. This is the case for unitary dynamics of finitely dimensional quantum systems for which it was proven \cite{Eisert2011} that its equilibrium states are Gibbs and von Neumann states described by a maximum entropy constrained with expectation values of integrals of motion. Our approach and results complement and take a different angle on problems studied by Burgarth and coauthors \cite{Bu1,Bu2}. 

The main aim of the paper is to provide answers to all above questions and consistently formulate the Jaynes' principle for equilibrium  and more generally asymptotic trajectories of all discrete and continuous quantum Markov trace-preserving processes (QMPs) governing the evolution of finite dimensional quantum systems. This is achieved exploiting algebraic properties of attractors and consequently of asymptotic states of QMPs developed in \cite{Albert2014,Novotny2018}. We derive three versions of Jayne's principles addressing three different settings. The first one concerns the situation when the initial state and thus all expectation values of observables are known. The second principle captures problem when the system is known to be sufficiently near the asymptotic regime and expectation values of only some constants of motion are known. The third version is designed for situation typical in statistical physics when quantum system is supposed to be sufficiently near to some equilibrium state and we know expectation values of some integrals of motion only.

The structure of the paper is as follows. In section \ref{sec:QMPs} we introduce dynamics of discrete and continuous quantum Markov processes, properties of their asymptotic dynamics relevant for this study, and different exponential forms of their asymptotic trajectories. Section \ref{sec:formulation_of-principles}) is devoted to formulations and proofs of all presented Jaynes' principle. Consequent properties of corresponding partition function and entropic relations are derived in section \ref{sec:statistical_properties}. We conclude in section \ref{sec:conclusion}.

\section{Quantum Markov processes and its asymptotics}
\label{sec:QMPs}
In this part we introduce quantum Markov processes and theory for their asymptotic evolution. Throughout the paper, we assume a quantum system with a finitely dimensional Hilbert space $\hil H$ of dimension $N$ and associated Hilbert space of all operators $\set B(\hil H)$ equipped with the Hilbert-Schmidt scalar product $(A,B)=\Tr\{A^{\dagger}B\}$ for $A,B \in \set B(\hil H)$. The state of the system is described by a density operator $\rho \in \set B(\hil H)$ which is positive with trace one.

\subsection{Quantum Markov dynamics}
We focus on the uniform quantum Markovian dynamics, which is memoryless and its evolution depends solely on the actual state of the system. Two types of  Markovian dynamics are distinguished -- quantum Markov chains (QMCHs) proceeding in discrete steps given by propagator $\map T_n = (\map T)^n$ for $n \in \Naturals_0$ steps  \cite{Gudder2008} and continuous quantum Markov dynamical semigroups (QMDSs) with propagator $\map T_t= \exp(\map L t)$ generated by Lindbladian $\map L$ with $t \in \Reals_0^{+}$ being continuous time \cite{Alicki2007}. Generators $\map T$ and $\map L$ act on the space $\set B(\hil H)$.
Consequently, during both types of Markovian evolution, the corresponding propagator $\map T_s$ shifts the state of a quantum system $\rho(t)$ at some time $t$ to the state $\rho(s+t) = \map T_{s} \rho(t)$ at time $s+t$. A properly defined quantum evolution requires completely positive trace-preserving propagators (quantum channels) $\map T_{s}$. Thus, while the one step generator $\map T$ admits the decomposition
\begin{equation}
\rho(n+1) = \map T(\rho(n)) = \sum_i K_j \rho(n) K_j^{\dagger},
\end{equation}
with Kraus operators $K_j$ satisfying $\sum_j K_j^{\dagger}K_j = I$ \cite{Chuang2010}, any Lindbladian $\map L$ must be of the form
\begin{equation}
\label{lindbladian_form}
\map L(X)=i[X,H] + \sum_j \left(L_j X L_j^{\dagger} - \frac{1}{2} \left\{L_j^{\dagger}L_j,X \right\}\right),
\end{equation}
with a Hamiltonian $H$ and Lindblad operators $L_j$ \cite{Alicki2007}. 
A special role play unital QMCHs and QMDSs whose evolution preserve the maximally mixed state $\frac{1}{N} I$, i.e. they must follow $\map T_s(I)=I$, where $I$ stands for the identity operator.

Analogously we can describe quantum Markov evolution of observables in the Heisenberg picture, in which mean values of an evolved observable $B(t)$ have to satisfy 
$$\Tr\left\{B(0)\rho(t)\right\}=\<B(0)\>_{\rho(t)}=\<B(t)\>_{\rho(0)}=\Tr\left\{B(t)\rho(0)\right\}.$$ 
Thus in the Heisenberg picture, discrete QMCH and continuous QMDS is generated by the unital adjoint map $\map T^{\dagger}$ satisfying $\map T^{\dagger}(I)=I$ and $\map L^{\dagger}$ respectively.

\subsection{Asymptotic dynamics of QMPs}
In this work we concentrate on the asymptotic dynamics of QMPs in the Schr\"{o}dinger and Heisenberg picture \cite{Novotny2018}. The former includes stationary states as well as periodic or quasi-periodic trajectories of a given QMP. All these asymptotic states live on the subspace $P\hil H$ \cite{Novotny2012,Amato2023}, i.e. their supports belong to the subspace $P\hil H$, where so-called $\map T$-projector $P$ is the projection onto the support of the maximal invariant state $\sigma$ of the given quantum Markov process. This invariant state with the maximal support is called the $\map T$-state. The $\map T$-state plays a central role in investigation of the asymptotic dynamics. In the case of unital QMPs, the  $\mathcal{T}$-state is the maximally mixed state whose support is the whole Hilbert space $\hil H$ and thus $P=I$. In a general QMP, any $\map T$-state can be, in principle, constructed as the time-averaged state
\begin{equation}
\label{sigma_def}
    \sigma=\lim\limits_{t\rightarrow\infty}\frac{1}{t}\SumInt_{0}^{t}ds\map T_{s}(\rho(0))
\end{equation}
of any strictly positive initial state $\rho(0)$. We stress that $\mathcal{T}$-states are not unique in general. For many QMPs, there is a freedom in choice of the $\map T$-state $\sigma$. Similarly, an asymptotic trajectory $\sigma(t)$ whose all states have its support equal to $P\hil H$ is called a $\map T$-trajectory. Note that a $\map T$-state is a special example of a $\map T$-trajectory. Any $\map T$-trajectory can be constructed as the asymptotic part of the evolution of some strictly positive initial state.

The asymptotic dynamics of QMPs runs inside the operator subspace $P\set B(\hil H)P$ and conveniently can be expressed in terms of its attractors in the Schr\"{o}dinger picture. Let us first treat both types of QMPs separately. We start with QMCHs. Attractors of a QMCH generated by $\map T$ are all operators from the so-called attractor (or asymptotic) space $\set{Atr}(\map T)$ defined as
\begin{equation}
\label{def_attractor_space}
\set{Atr}(\map T) = \bigoplus_{\lambda \in \sigma_{as}}\set{Ker}(\map T - \lambda I),
\end{equation}
\noindent where the asymptotic spectrum $\sigma_{as}$ is
\begin{equation}
    \sigma_{as}=\{\lambda| \lambda \in \sigma(\map T),|\lambda|=1\},  \nonumber
\end{equation}
with  $\sigma(\map T)$ denoting the spectrum of the generator $\map T$. In other words, the attractor space is the span of eigenoperators associated with eigenvalues from the asymptotic spectrum. Let us stress that attractors are, in general, not states (density operators). These are operators from which we can construct the asymptotic evolution of the given QMCH. Indeed, equipped with a linearly independent basis $X_{\lambda,j}$ of individual attractor subspaces corresponding to $\lambda\in\sigma_{as}$ and their dual basis $X^{\lambda,j}$, i.e. $\left( X^{\lambda,j},X_{\mu,l}\right)=\delta_{\lambda \mu}\delta_{jl}$, the asymptotic trajectory of an initial state $\rho(0)$ reads
\begin{equation}
\label{eq_asymptotic_dynamics}
\rho(t\gg 1) = \sum_{\lambda \in \sigma_{as},j} \lambda^{t} X_{\lambda,j}\Tr\left[X^{\lambda,j\dagger}\rho(0)\right],
\end{equation}
where $j$ accounts for possible degeneracies of eigenoperators  \cite{Novotny2018}. 

Similarly, attractors of a QMDS generated by Lindbladian $\map L$ are all operators from the so-called attractor (or asymptotic) space $\set{Atr}(\map L)$ defined in this case as
\begin{equation}
\label{def_attractor_space_2}
\set{Atr}(\map L) = \bigoplus_{a \in \sigma_{as}(\map L)}\set{Ker}(\map L - a I),
\end{equation}
\noindent where the asymptotic spectrum $\sigma_{as}(\map L)$ is
\begin{equation}
    \sigma_{as}(\map L)=\{a| a \in \sigma(\map L), \set{Re}(a)=0\}, \nonumber
\end{equation}
with  $\sigma(\map L)$ denoting spectrum of generator $\map L$. Similarly, considering linearly independent basis  $X_{a,j}$ and their corresponding dual basis $X^{a,j}$, the asymptotic trajectory of an initial state $\rho(0)$ reads \cite{Novotny2018}
\begin{equation}
\label{eq_asymptotic_dynamics_2}
\rho(t\gg 1) = \sum_{a \in \sigma_{as}(\map L),j} \exp(at) X_{a,j}\Tr\left[X^{a,j\dagger}\rho(0)\right],
\end{equation}
where $j$ accounts for possible degeneracies of eigenoperators.

Since we intend to treat both types of QMPs simultaneously and in a unified manner, we formally introduce a one step generator $\map T=\map T_1=\exp(\map L)$ of evolution for QMDS generated by the Lindbladian $\map L$. To avoid eventual ambiguity of its operator powers we define $\map T^t= \exp(\map L t)= \map{T}_t$ for $t \in \Reals_{+}$. Concerning asymptotic dynamics, its attractor space $\set{Atr}(\map T)$ coincides with the attractor space $\set{Atr}(\map L)$. However, we must also keep the relationship between any attractor $X_{\lambda,j} \in \set{Atr}(\map T)$ and its associated $a \in \sigma_{as}(\map L)$ via relation $\lambda = \exp(a)$. In other words, attractors corresponding to the same $\lambda$ behave differently under evolution if they are associated with different $a_1,a_2 \in \sigma_{as}(\map L)$ satisfying 
\begin{equation}
\lambda= \exp(a_1)=\exp(a_2).
\end{equation}
For example, while fixed points of QMCHs are attractors corresponding to $\lambda=1$, fixed points of QMDSs are attractors corresponding to $\lambda=\exp(a)=1$ with only $a=0$. In the rest of the paper we will keep the following notation: we denote attractor spaces of both types of QMPs generated by one step generator $\map T$ as $\set{Atr}(\map T)$ and their asymptotic spectra as $\sigma_{as}$, which allows us to write the asymptotic dynamics of both types of QMPs as (\ref{eq_asymptotic_dynamics}).

Analogously we define the attractor space $\set{Atr}(\map T^{\dagger})$ of the QMP in the Heisenberg picture. Its basis are formed by the dual basis $X^{\lambda,j}$. In general these two attractor spaces do not coincide, but there are important mutual relationships between both sets of attractors \cite{Novotny2018} and consequently, as we will see, between asymptotic states and conserved quantities of the given QMP. Indeed, while the attractor space $\set{Atr}(\map T)$ contains all asymptotic trajectories, the attractor space $\set{Atr}(\map T^{\dagger})$ contains all constants of motion \cite{Albert2014}. Moreover the projected part of the attractor space $P\set{Atr}(\map T^{\dagger})P$ forms a $C^{*}$-algebra, i.e. with each pair of operators it contains also their product, their adjoint operators and it is nonempty since it contains at least the projection $P$. We note that the whole $\set{Atr}(\map T^{\dagger})$ is also a $C^{*}$-algebra but under the so-called Choi-Effros product \cite{Bhat2022}.

\subsection{Exponential form of asymptotic states}
\label{sec:exponential_form_of_states}
All these attractor properties listed in the previous part provide us with an exponential form of all asymptotic states of QMPs.
It was shown that any strictly positive asymptotic state $\rho_{as}$ on $P\hil H$ (\ref{eq_asymptotic_dynamics}) can be written as \cite{Novotny2018}
\begin{equation}
\label{eq_exponential form}
\rho_{as}= \exp\left[\ln \sigma - \sum_j \gamma_j PZ_jP \right],
\end{equation}
where $\{Z_j\}$ is a hermitian basis of the attractor space $\set{Atr}(\map T^{\dagger})$, $\gamma_j$ are real parameters and $\sigma$ is any $\map T$-state. Please note that the state $\sigma$ is strictly positive on the subspace $P\hil H$ and vanishing on its orthogonal complement. Thus the logarithm of the state $\sigma$ is given by the logarithm of its strictly positive part while being zero on its orthogonal part. All the other asymptotic states, which are not strictly positive, are obtained from (\ref{eq_exponential form}) simply by taking the limit $+\infty$ or $-\infty$ for some chosen subset of parameters $\gamma_j$ \cite{Novotny2018}.

Consequently, any discrete or continuous asymptotic trajectory (\ref{eq_asymptotic_dynamics}) can be written as (\ref{eq_exponential form}) with time-dependent coefficients. If all these parameters are finite, it describes an asymptotic trajectory $\omega(t)$ whose each state has the maximal support $P\hil H$
\begin{equation}
\label{eq_exponential form_time}
\omega(t)= \exp\left[\ln \sigma - \sum_j \mu_j(t) PZ_jP \right].
\end{equation}
It is a $\map T$-trajectory. Expressing 
$\ln \sigma = \ln \omega(t) + \sum_j \mu_j(t) PZ_jP$
we can write down any asymptotic trajectory $\rho_{as}(t)$ as
\begin{eqnarray}
\label{eq_exponential form_shifted}
\rho_{as}(t)&=& \exp\left[\ln \sigma - \sum_j \nu_j(t) PZ_jP \right] \nonumber \\ &=& \exp\left[\ln \omega(t) - \sum_j \gamma_j(t) PZ_jP \right].
\end{eqnarray}
It has two consequences. First, the role of the $\map T$-state in the exponential form of asymptotic states can be replaced by the more general $\map T$-trajectory. We will find it very useful later when we formulate one of Jaynes' principles. Second, if we choose a different $\map T$-trajectory it only shifts the original values of parameters $\nu_j(t)$ to $\gamma_j(t)=\nu_j(t) + \mu_j(t)$.

We, instead of analyzing the time dependence of parameters $\gamma_j(t)$ determining asymptotic trajectories, provide a natural description of asymptotic trajectories via constants of motion of QMPs, which also allow to formulate the Jaynes' principle for all QMPs. A time dependent observable $C(t)=C(t)^{\dagger} \in \set{B}(\hil H)$ is called constant of motion if 
\begin{equation}
\label{def_constant_of_motion}
\<C(t+s)\>_{\rho(t)} = \Tr\left[C(t+s)\rho(t)\right] = \<C(s)\>_{\rho(0)}
\end{equation}
for any state trajectory $\rho(t)$, i.e. its expectation value remains constant along all state trajectories irrespective of the reference time $s$. In this definition we take into account the fact that, due to open dynamics, some system degrees of freedom may become suddenly unavailable, leaving the possibility to define constants of motions freely in this redundant dying part. This occurs, for example, in the case of discrete QMP that has, apart of the asymptotic spectrum, generalized eigenvectors corresponding to eigenvalue $0$ only. Since these can be of different ranks, the set of initial states is gradually narrowed down by the evolution until it reaches its asymptotic regime. Outside of this dying part, e.g. for sufficiently long time outside of the space $P\hil H$, a constant of motion can be defined freely, since it has no additional contribution to its mean value. Requiring time homogeneity for constants of motions naturally removes this redundant part. Indeed, the definition (\ref{def_constant_of_motion}) may be rewritten as
\begin{equation}
\Tr\left[T_t^{\dagger}\left(C(t+s)\right)\rho(0)\right] = \Tr\left[C(s)\rho(0)\right]  \nonumber
\end{equation}
valid for any initial state $\rho(0)$. It directly implies that any constant of motion $C(t)$ follows
\begin{equation}
C(t_2-t_1)=\map T_{t_1}^{\dagger}(C(t_2))
\end{equation}
for any $0 \leq t_1 \leq t_2$. Thus the constants of motion are hermitian (observable) trajectories of QMPs in the Heisenberg picture running in the reversed time direction within the attractor space. It implies the set $\map C$ of all constants of motion of a given QMP forms a finite-dimensional real vector space. Indeed, one can easily check that the basis of all asymptotic trajectories can be simply chosen as $\{\lambda^t X^{\lambda,j}\}$ ($\lambda \in \sigma_{as}$). Since attractors always come in mutually conjugated pairs associated with complex conjugated eigenvalues, we can choose basis of $\map C$ as 
\begin{eqnarray}
C_{\lambda,j}^{(+)}(t)&=&\frac{1}{2}\left(\lambda^{t}X^{\lambda,j}+\overline{\lambda^{t}}\left(X^{\lambda,j}\right)^{\dagger}\right), \nonumber \\
C_{\lambda,j}^{(-)}(t)&=&\frac{1}{2i}\left(\lambda^{t}X^{\lambda,j}-\overline{\lambda^{t}}\left(X^{\lambda,j}\right)^{\dagger}\right),
\end{eqnarray}
for each $\lambda \in \sigma_{as}$ and $j$. Note again that powers of eigenvalues from asymptotic spectrum of QMDSs are defined via $\lambda^t=\exp(at)$. Consequently, the vector space of all constants of motion $\map C$ has the same dimension as the attractor spaces, i.e. 
$$\dim(\map C)= \dim(\set{Atr}(\map T))=\dim(\set{Atr}(\map T^{\dagger}).$$ 
Of course, the set $\map C$ contains also real subspace $\map I$ of all integrals of motion of the QMP. These are time independent constants of motion, which always contain the identity operator $I$, since we assume only trace-preserving QMPs keeping $\Tr[\rho(t)]=1$.

Thus, any element $Z_i$ of hermitian basis of $\set{Atr}(\map T^{\dagger})$ employed in (\ref{eq_exponential form_shifted}) can be expressed in terms of constants of motion. Let us assume that apart from the identity operator there are $d$ linearly independent constants of motion and denote elements of the basis as $\{I,C_1(t),C_2(t),\ldots,C_d(t)\}$. Any asymptotic trajectory can be written as 
\begin{eqnarray}
\label{eq_exponential form_general}
\rho_{as}(t)&=& \exp\left[\ln \sigma(t) - \sum_j \gamma_j(t) PC_j(t)P-\alpha(t) P \right] \nonumber \\
&=& \frac{1}{\map Z}\exp\left[\ln \sigma(t) - \sum_j \gamma_j(t) PC_j(t)P \right],
\end{eqnarray}
where we used $PIP=P$ and the fact that $P$ commutes with all operators of the form $PXP$. The parameter $\alpha(t)$ is associated with the integral of motion $I$ and $\mathcal{Z}=exp(\alpha(t))$ is a normalization function. It depends on parameters $\gamma_i(t)$ and ensures $\Tr[\rho(t)]=1$.

We found the form of all asymptotic trajectories including stationary states as well, which very clearly resembles states (\ref{eq_equilibrium_state}). It is based on properties of attractors in both pictures. However, we stress that states of these trajectories do not maximize the von Neumann entropy in general and moreover the form covers also non-stationary asymptotic trajectories. Equipped with the general form of asymptotic states (\ref{eq_exponential form_general}) of QMPs expressed in terms of their constants of motion, it is our next main task to formulate properly the Jaynes' principle for QMPs, i.e. quantum information principle which would equivalently reproduce all asymptotic trajectories (\ref{eq_exponential form_general}) of QMPs.

\section{Formulation of Jaynes' principles}
\label{sec:formulation_of-principles}

In order to write down all asymptotic trajectories, we have to know at least one $\map T$-trajectory, e.g. one $\map T$ state. A special role among them plays the trajectory initiated at the maximal mixed state and its time-averaged $\map T$-state. The following lemma gathers their commutation properties which in turn will be found crucial in formulation of Jaynes' principles.

\begin{lemma}[The maximally mixed state]
\label{lemma_max_mixed_state}
Let $\rho_I(t) = \map T_t\left(\frac{1}{N}I\right)$ be the trajectory of a QMP $\mathcal{T}_{t}$ initiated in the maximally mixed state $\frac{1}{N}I$, $\sigma_I(t)=\rho_I(t \gg 1)$ its asymptotic part (asymptotic trajectory) and $\sigma_I$ its time-averaged state. Then 
\begin{itemize}
\item[1.] $\sigma_I(t)$ is a $\map T$-trajectory and its time average $\sigma_I$ is a $\map T$-state.
\item[2.] All states of $\rho_I(t)$ and thus $\sigma_I(t)$ and $\sigma_I$ commute with an arbitrary constant of motion of the given QMP projected by the corresponding $\mathcal{T}$-projector $P$. Consequently, $\sigma_I(t)$ and $\sigma_I$ commute also with all asymptotic states of the given QMP.
\end{itemize}
\end{lemma}
\begin{proof}
1. $I \geq \sigma$ for any state $\sigma$ implies that $\map T_{t}(I) \geq \map T_{t}(\sigma)$. Thus $\rho_I(t)$ has at any time $t$ the highest rank from all evolved states and thus in the asymptotic regime it is a $\map T$-trajectory, denoted as $\sigma_I(t)$. Its time-averaged state $\sigma_I$ is by construction a stationary state of a QMP with the maximal rank, i.e. $\sigma_I$ is a $\map T$-state.

2. To prove the second statement, let $A_k(t)$ be the Kraus operators of the QMP $\mathcal{T}_{t}$. Applying their commutation properties  
\begin{equation}
    \begin{aligned}
    A_k(t)PX^{\lambda,i}P &= \lambda^t PX^{\lambda,i}PA_k(t),\\
    PX^{\lambda,i}PA_k(t)^{\dagger}&=\lambda^{t} A_k(t)^{\dagger} PX^{\lambda,i}P
    \end{aligned}
\end{equation}
(see Theorem 6 in \cite{Novotny2018}) to basis of constants of motion chosen in the form $\{\lambda^{t} X^{\lambda,j} + \overline{\lambda}^{t} X^{\lambda,j \dagger}, i(\lambda^{t} X^{\lambda,j} - \overline{\lambda}^{t} X^{\lambda,j \dagger})\}$ we receive general relations
\begin{eqnarray}
\label{eq_general_identity}
\mathcal{T}_{s}(PC(t)PB) &=& PC(t+s)P \mathcal{T}_{s}(B) \nonumber \\ 
\mathcal{T}_{s}(BPC(t)P) &=& \mathcal{T}_{s}(B) PC(t+s)P
\end{eqnarray}
for any constant of motion $C(t)$ and operator $B$. Choosing $B= \frac{1}{N}I$ we arrive at the desired relation
\begin{equation}
\label{eq_idencom}
PC(t+s)P\rho_I(t)=\rho_I(t)PC(t+s)P.
\end{equation}
Naturally, the same commutation relation is satisfied by states from the $\mathcal{T}$-trajectory $\sigma_I(t)$. Now, since $\sigma_I(t+s)$ is also a $\mathcal{T}$-trajectory for any reference time $s$, any asymptotic trajectory $\rho(t)$ (\ref{eq_exponential form_general}) can be rewritten as
\begin{equation}
    \rho(t)= \exp\left[\ln \sigma_I(t+s) - \sum_j\gamma_j(t) PC_j(t)P -\alpha(t) P\right]. \nonumber
\end{equation}
Combining this form with (\ref{eq_idencom}) we finally receive
\begin{equation}
\label{commutation}
   \sigma_I(t+s)\rho(s)=\rho(s)\sigma_I(t+s)
\end{equation}
for any positive $t$ and $s$, indicating that states from the asymptotic trajectory of the maximally mixed state commute with all asymptotic states.
\end{proof}

The final conceptual piece we need in our quest for a generalized Jaynes principle is a properly chosen functional to be extremized. As we will see, the quantum relative entropy $S(\rho_1|\rho_2)$ of the state $\rho_1$ with respect to the state $\rho_2$ defined as 
\begin{eqnarray*}
S(\rho_1|\rho_2) = \left\{\begin{array}{ll}
\Tr\left[\rho_1 (\ln \rho_1 - \ln \rho_2) \right] & \mbox{if $\Supp(\rho_1) \subseteq \Supp(\rho_2)$} \\ +\infty & \mbox{otherwise}. 
\end{array} \right.
\end{eqnarray*}
is the needed object. Equipped with the knowledge of asymptotic states (\ref{eq_exponential form_general}) we formulate the first Jaynes' principle for a fully known initial state evolved under a quantum Markov process.

\begin{theorem}[Jaynes's principle for fully known initial states]
\label{theorem_jaynes}
Assume a QMP $\map T_{t}$ with a $\map T$-trajectory $\sigma(t)$. Let $P$ be the $\map T$-projector and $\{I,C_1(t), \ldots, C_d(t)\}$ a basis of constants of motion. Let $\rho(0)$ be the initial state with expectation values $c_j=\<C_j(t)\>_{\rho(t)}$. Then individual states 
\begin{equation}
\label{eq_gibbs form}
\rho(t)= \frac{1}{\map Z}\exp\left[\ln \sigma(t) - \sum_j \gamma_j PC_j(t)P \right],
\end{equation}
from its asymptotic trajectory (\ref{eq_asymptotic_dynamics}) minimize, under the given expectation values, the quantum relative entropy $S(\rho(t)|\sigma(t))$ for any sufficiently long time $t$. The Lagrange multipliers $\gamma_j$ are time independent constants regardless of a given QMP and are determined by expectation values $c_j$. The partition function $\mathcal{Z}$ is a function of the Lagrange multipliers $\gamma_j$ and ensures $\Tr[\rho(t)]=1$.
\end{theorem}
\begin{proof}
Let us assume an arbitrary initial state $\rho(0)$. In the following we show that the Jaynes' principle formulated in Theorem \ref{theorem_jaynes} provides exactly the same asymptotic trajectory (\ref{eq_gibbs form}). The quantum relative entropy (\ref{eq_quantum_relative_entropy}) is finite if and only if $\rho(t) = P\rho(t)P$. Thus, looking for the minimum of $S(\rho(t)|\sigma(t))$ we can restrict ourselves to states supported on the subspace $P\hil H$. This modifies known expectation values of constants of motion into constrains
\begin{eqnarray}
\<C_j(t)\>_{\rho(t)} &=& \Tr\{C_j(t)P\rho(t)P\}=\Tr\{PC_j(t)P\rho(t)\}=c_j, \nonumber \\
\<I\>_{\rho(t)} &=& \Tr\{P\rho(t)\}=1.
\label{app_constrains}
\end{eqnarray}
We seek a constrained minimum of $S(\rho(t)|\sigma(t))$ at any asymptotic time $t$. Following the standard procedure for constrained extremes we introduce the Lagrange functional
\begin{equation}
    \Lambda = \Tr\left[\rho(t) \left(\ln \frac{\rho(t)}{\sigma(t)} + \sum_j \gamma_j(t) PC_j(t)P + \alpha(t) P \right)\right], \nonumber
\end{equation}
which depends on Lagrange multipliers $\gamma_j(t)$ and $\alpha(t)$ mutually associated with each constrain in (\ref{app_constrains}). Its variation $ \delta\Lambda$, i.e. the first order leading linear contribution,
\begin{equation}
    \delta\Lambda = \Tr\biggl[\delta\rho(t) \biggl(\ln \frac{\rho(t)}{\sigma(t)}+(1+\alpha(t))P + \sum_j \gamma_j(t) PC_j(t)P \biggr)\biggr] \nonumber
\end{equation}
must be equal to zero for any hermitian variation $\delta\rho(t)$ of the extremal state $\rho(t)$. We find that the only solution of this condition at a given asymptotic time $t$ is exactly the state (\ref{eq_gibbs form}). Since the prime role of $\map Z(t) = \exp[1+\alpha(t)]$ is to normalize the trace of the state $\rho(t)$ to one, it coincides with the definition given in (\ref{eq_gibbs form}) and we call it the partition function.

In the following we show that the Lagrange multipliers are time-independent constants regardless of a given QMP. According to (\ref{eq_exponential form_shifted}), changing a $\mathcal{T}$-state (trajectory) manifests itself only by shifting the parameters $\gamma_i(t)$ and we are therefore free to choose an arbitrary $\mathcal{T}$-trajectory. Let us first exploit the $\mathcal{T}$-trajectory $\sigma_I(t)$ and apply (\ref{eq_general_identity}) to evolution of an arbitrary asymptotic trajectory $\rho(t)$
    \begin{eqnarray*}
        &&\rho(t+s)=\mathcal{T}_s(\rho(t)) = \nonumber \\
        && \frac{1}{\map Z(t)}\mathcal{T}_s\left(\exp\left[-\sum_j\gamma_j(t)PC_j(t)P\right]\sigma_I(t)\right) \nonumber \\ &=&
        \frac{1}{\map Z(t)}\sum_{N=0}^{\infty}\frac{(-1)^N}{N!}\mathcal{T}_s\left(\left(\sum_j\gamma_j(t)PC_j(t)P\right)^N\sigma_I(t)\right) \nonumber \\&=&
        \frac{1}{\map Z(t)}\sum_{N=0}^{\infty}\frac{(-1)^N}{N!}\left(\sum_j\gamma_j(t)PC_j(t+s)P\right)^N\mathcal{T}_s(\sigma_I(t))\\&=&
        \frac{1}{\map Z(t)}\exp\left[\ln\sigma_I(t+s)-\sum_j\gamma_j(t)PC_j(t+s)P\right],
\end{eqnarray*}
which implies $\gamma_j(t+s)=\gamma_j(t)$, $\map Z(t+s)=\map Z(t)$ and therefore the values of the Lagrange multipliers associated with the $\mathcal{T}$-trajectory $\sigma_I(t)$ are fixed by their initial values $\gamma_j(t)=\gamma_j(0)\equiv\gamma_j$, $\map Z(t)=\map Z(0)\equiv \map Z$. The same reasoning applies to Lagrange multipliers for a general $\mathcal{T}$-trajectory $\sigma(t)$ associated with the $\mathcal{T}$-trajectory $\sigma_I(t)$ and we can finally conclude that the asymptotic trajectory $\rho(t)$ takes the form (\ref{eq_gibbs form}).
\end{proof}
A natural option is to choose $\sigma(t)$ as some $\mathcal{T}$-state, but as we will see, by keeping  the general $\mathcal{T}$-trajectory in the theorem will allow us to formulate the Jaynes' principle for incompletely known asymptotic states.

Let us focus on important details and consequences concerning the theorem. First of all, we prove that despite the fact that the expectation values constitute time dependent constrains, Lagrange multipliers $\gamma_j$ are time-independent. Their physical meaning is prescribed solely by expectation values of constants of motion of the given QMP and can not be, as expected, revealed by the Jaynes' principle itself. Second, the theorem determines all asymptotic trajectories (\ref{eq_asymptotic_dynamics}) including those with non-strictly positive states. Lagrange multipliers can take arbitrary real values and non-strictly positive asymptotic states emerge when some of these parameters are taken in the limit $\pm \infty$. It is actually a well known situation in statistical physics, e.g. in the canonical ensemble, ground states are obtained in the limit $\beta=1/T$ ($T$ denoting a temperature of the canonical system) going to $+\infty$ \cite{Balian}. Third, the Theorem \ref{theorem_jaynes} states that expectation values of all constants of motion are the only information required in the Jaynes' principle to determine individual asymptotic trajectory of any initial state. From  (\ref{eq_asymptotic_dynamics}) with $\sigma(t)=\sigma$, one sees that this asymptotic trajectory corresponds to a stationary state if and only if all expectation values of non-stationary constants of motion are zero. Thus, if additional information is provided stating that the resulting asymptotic state is stationary, the modified Jaynes' principle for asymptotically stationary states reduces to minimization of quantum relative entropy $S(\rho(t)|\sigma)$ only with respect to the expectation values of integrals of motion, $\sigma(t)=\sigma$ and all Lagrange multipliers corresponding to non-stationary constants of motion are set to zero in (\ref{eq_exponential form}).

This brings us to another interesting problem. In statistical physics, typically only some expectation values of constants of motion are known and they constitute our only available information about a given equilibrium state. Can we adopt the same approach also in this scenario and simply set Lagrange multipliers associated with unknown (or ignored) expectation values of constants of motion to zero? The answer is No. This leads to a contradiction which is easy to demonstrate on the case when we lack knowledge of all expectation values. Indeed, resetting all Lagrange multipliers in equivalent representations (\ref{eq_gibbs form}) with different choices of the $\map T$-state $\sigma$ results in an ambiguity of the asymptotic state. In fact, it is a very different task which, to be resolved, we must properly incorporate the given lack of knowledge via evolution into the description of asymptotic trajectories. We proceed in the following way. We stress, the initial state of the system is not known. In principle, it could be any state from trajectories consistent with the given knowledge. Thus also the initial time is not relevant. As all states evolve eventually into states living exclusively on the subspace $P\hil H$, we choose the initial state at $t=0$ being already supported on this subspace. Expectation values of constants of motion then read
\begin{equation}
    \label{conmo_inf}
    \<C_i(0)\>_{\rho(0)} =  \Tr \left\{\rho(0) PC_i(0)P \right\}= c_i \quad  i \in \set{K},
\end{equation}
where the set $\set K$ collects indices of linearly independent constants of motion whose expectation values are known. The state comprising solely the given constraints reads
\begin{equation}
\label{eq_initial_asymptotic_state}
    \rho(0) = \frac{1}{\map Z}\exp\left(-\sum_{i \in \set{K}} \gamma_i PC_i(0)P\right).
\end{equation}
However, this is not an asymptotic state yet in general. To get the asymptotic state with the prescribed expectation values we simply let it evolve towards its asymptotic trajectory, which is determined by the Jaynes' principle in the following form.

\begin{theorem}[Jaynes's principle for partially known asymptotic states]
\label{theorem_incomplete_Jaynes}
Consider $\map T_{t}$ be a QMP with $\sigma_I(t)$ being the asymptotic evolution of the maximally mixed state and $P$ the $\map T$-projector. Let the quantum system be in the asymptotic regime of the evolution and expectation values of some linearly independent constants of motion ($c_j=\Tr\{\rho(t)C_j(t)\}$ with $j \in \set K$) be the only knowledge provided about its asymptotic trajectory. Then states $\rho(t)$ from the corresponding asymptotic trajectory minimize, under the given expectation values, quantum relative entropy $S(\rho(t)|\sigma_I(t))$ for any sufficiently long time $t$. The asymptotic trajectory takes the form of generalized Gibbs and von Neumann state
\begin{equation}
        \label{eq_gibbs_states_incomplete}
    \rho(t) = \frac{1}{\map Z}\exp\left(\ln \sigma_I(t) - \sum_{j \in \set{K}} \gamma_j PC_j(t)P\right),
\end{equation}
where the partition function $\map Z$ of time-independent Lagrange multipliers $\gamma_j$ ensures $\Tr \left[\rho(t)\right]=1$.
\end{theorem}

\begin{proof}
Suppose our knowledge about the initial state is reduced to information about expectation values of some of constants of motion $C_j(t)$ with $j\in\textsf{K}$, supplemented by the knowledge that the system is known to be already in the asymptotic regime.  At time $t=0$ the state supported on the subspace $P\mathscr{H}$ with the known expectation values (\ref{conmo_inf}) is given by (\ref{eq_initial_asymptotic_state}). However, this state is not asymptotic yet. To satisfy the second requirement we have to let it evolve to its asymptotic trajectory. According to (\ref{eq_general_identity}) and Lemma \ref{lemma_max_mixed_state}, the evolution of such state reads
\begin{eqnarray*}
\rho(t)&=&\mathcal{T}_{t}(\rho(0))=\frac{1}{\map Z}\mathcal{T}_{t}\left(\frac{1}{N}I\right)\exp\left[-\sum_{j\in\textsf{K}}\gamma_jPC_j(t)P\right]\\ &=&
    \frac{1}{\map Z}\exp\left[\ln\rho_I(t)-\sum_{j\in\textsf{K}}\gamma_jPC_j(t)P\right].
\end{eqnarray*}
After a sufficiently large time $t$, we obtain the asymptotic trajectory $\rho_{\infty}(t)$
\begin{equation}
    \rho_{\infty}(t)=\frac{1}{\map Z}\exp\left[\ln\sigma_I(t)-\sum_{j\in\textsf{K}}\gamma_jPC_j(t)P\right]. \nonumber
\end{equation}
It is the asymptotic trajectory of the given QMP representing the given knowledge (\ref{conmo_inf}). According to Theorem 1, such a quantum trajectory minimizes the relative quantum entropy $S(\rho(t)|\sigma_I(t))$ with respect to the given knowledge about the asymptotic state.
\end{proof}

We stress that despite the seeming similarity of Theorems \ref{theorem_jaynes} and \ref{theorem_incomplete_Jaynes} they treat different situations. While Theorem \ref{theorem_jaynes} deals with asymptotic evolution of the given initial state, Theorem \ref{theorem_incomplete_Jaynes} concerns quantum systems being in the asymptotic regime whose expectation values of some constants of motion are known. Since we do not know the system's initial state and consequently the time when the dynamics of the system enters into its asymptotic regime, the Jaynes' principle formulated in Theorem \ref{theorem_incomplete_Jaynes} determines the asymptotic trajectory, but it can not provide the actual position of the quantum system on this asymptotic trajectory. Another difference lies in treatment of parameters $\gamma_j$. For a specific choice of the $\map T$-trajectory $\sigma(t)$ in Theorem \ref{theorem_jaynes}, some of the parameters $\gamma_j$ may be equal to zero and thus the corresponding constants of motion $C_j(t)$ do not appear in the Gibbs form (\ref{eq_gibbs form}). This is much different to the situation treated by Theorem \ref{theorem_incomplete_Jaynes}, in which the $\map T$-trajectory is uniquely specified. Constants of motion $C_j(t)$, $j\notin\set{K}$, are completely ignored in the Gibbs and von Neumann state (\ref{eq_gibbs_states_incomplete}) and thus they do not have assigned any parameters $\gamma_j$.

Theorem \ref{theorem_incomplete_Jaynes} reveals the unique role of the evolution of the maximally mixed state among all possible trajectories. At this point let us stress that even in the case, when only expectation values of integrals of motion are known, stationarity of the resulting asymptotic trajectory is not guaranteed. It is ensured only if the maximally mixed state is steady too. But how to treat the situation when we know that the quantum system has reached a stationary asymptotic state, or the asymptotic dynamics is fast enough that we observe its averaged stationary part solely, and additionally we know expectation values of only some of its integrals of motion? Again, we must integrate both types of information into the resulting state consistent with evolution of the given QMP. The following theorem deals with this issue.
\begin{theorem}[Jaynes's principle for stationary states with incomplete knowledge]
\label{theorem_stationary_incomplete}
Let $\map T_t$ be a QMP and $\sigma_I$ be a its time-averaged trajectory of the maximally mixed state. Suppose the system is in a partially unknown stationary state $\rho$ of the QMP and expectation values $c_j=\Tr\{\rho I_j\}$ of some linearly independent integrals of motion $I_j$, $j \in \set K$, be the only available knowledge about this state. Then the state $\rho$ minimizes, under the given expectation values, quantum relative entropy $S(\rho|\sigma_I)$ and takes the form of generalized Gibbs and von Neumann state
\begin{equation}
        \label{eq_stationary_gibbs_states_incomplete}
    \rho = \frac{1}{\map Z}\exp\left(\ln \sigma_I - \sum_{j \in \set{K}} \gamma_j PI_jP\right),
\end{equation}
where the partition function $\map Z$ of time-independent Lagrange multipliers $\gamma_j$ ensures $\Tr \left[\rho\right]=1$.
\end{theorem}
\begin{proof}
If the system is in some stationary state, it is in the asymptotic regime already. In view of Theorem \ref{theorem_incomplete_Jaynes}, the system in the asymptotic regime being determine solely by expectation values of some integrals of motion evolves as 
\begin{equation}
\label{eq_nonstationary_incomplete}
    \rho(t) = \frac{1}{\map Z}\exp\left(\ln \sigma_I(t) - \sum_{j \in \set{K}} \gamma_j PI_jP\right), \nonumber
\end{equation}
where $\sigma_I(t)$ is asymptotic $\map T$-trajectory initiated at the maximally mixed state. All these asymptotic states have the same demanded expectation values of given integrals of motion, but it is not stationary as it is required. To get a stationary state we have to make time-average of the trajectory $\rho(t)$ given above. Averaging can not increase amount of information imprinted into the state and consequently the resulting stationary state $\rho$ incorporates the minimum information given by expectation values of the corresponding integrals of motion. Exploiting the Lemma \ref{lemma_max_mixed_state}, we find that
\begin{eqnarray}
\rho &=& \lim\limits_{t\rightarrow\infty}\frac{1}{t}\SumInt_{0}^{t}ds \frac{1}{\map Z}\exp\left(\ln \sigma_I(s) - \sum_{j \in \set{K}} \gamma_j PI_jP\right) \nonumber \\
&=& \left(\underbrace{\lim\limits_{t\rightarrow\infty}\frac{1}{t}\SumInt_{0}^{t}ds \sigma_I(s)}_{\sigma_I}\right) \frac{1}{\map Z}\exp\left(- \sum_{j \in \set{K}} \gamma_j PI_jP\right) \nonumber \\
&=& \frac{1}{\map Z}\exp\left(\ln \sigma_I - \sum_{j \in \set{K}} \gamma_j PI_jP\right). \nonumber
\end{eqnarray}
According to Theorem \ref{theorem_jaynes}, this state minimizes under the expectation value constrains $c_j=\Tr\{\rho I_j\}$ with $j \in \set K$ the quantum relative entropy $S(\rho|\sigma_I)$.
\end{proof}

The last Theorem \ref{theorem_stationary_incomplete} treats exactly the setting well known in statistical physics --  equilibrium states characterized by few expectation values of integrals of motion. We find that all above formulated Jaynes' principles as well as the resulting general Gibbs and von Neumann forms of asymptotic states differs from the MaxEnt principle and associated Gibbs and von Neumann states (\ref{eq_equilibrium_state}) used in statistical physics. However, if the Markov evolution preserves the maximally mixed state (unital QMPs), then projection $P$ is the identity operator. The trajectory $\sigma(t)$ can be chosen as the fixed maximally mixed state and may be therefore included into the partition function and we get finally Gibbs form which for stationary states coincides with (\ref{eq_equilibrium_state}). From this perspective, Theorems \ref{theorem_jaynes}, \ref{theorem_incomplete_Jaynes} and \ref{theorem_stationary_incomplete}  prove the MaxEnt principle for equilibrium states of unital channels and generalizes it to Jaynes'principle for all asymptotic trajectories of any quantum Markov process.

\section{Properties of the partition function}
\label{sec:statistical_properties}
Surprisingly, despite the different generalized Gibbs and von Neumann form (\ref{eq_gibbs form}) of asymptotic states, the partition function keeps its special role in calculating statistical properties like expectation values of observables, quantum relative entropy and its change.
Indeed, the partition function $\map Z$ independently of the choice of a $\mathcal{T}$-trajectory provides the expectation values of constants of motion
\begin{equation}
\label{eq_mean_values_real}
\<C_j(t)\>_{\rho(t)} = - \left(\frac{\partial \ln \map Z}{\partial \gamma_j}\right)_{\gamma_{k \neq j}}.
\end{equation}
Proof of this formula exploits the fact that the partition function is a trace functional of the form $F(X)= \Tr f(X)$ on the space of all hermitian operators. If $f$ is an analytic function, the variation of the functional (the first order change in term of $\delta X$) reads $\delta F_X(\delta X)=\Tr\left\{ \delta X f^{'}(X)\right\}$. With the help of this expression we get for $X= \ln \sigma(t) - \sum_{k} \gamma_k PC_k(t)P$

\begin{eqnarray}
 &&\left(\frac{\partial \ln \map Z}{\partial \gamma_j}\right)_{\gamma_{k \neq j}} = \nonumber \\ && \frac{1}{\map Z}\lim_{\triangle \gamma_j \rightarrow 0} \frac{\Tr\left\{\exp\left(X - \triangle \gamma_j PC_j(t)P\right) \right\}- \Tr\left\{\exp\left( X \right)\right\}}{\triangle \gamma_j} = \nonumber \\
&&\frac{1}{\map Z}\lim_{\triangle \gamma_j \rightarrow 0} \frac{\Tr\left\{-\triangle \gamma_j PC_j(t)P\exp\left(X\right) \right\}}{\triangle \gamma_j} = -\<C_j(t)\>_{\rho(t)}. \nonumber
\end{eqnarray}

Different choices of $\mathcal{T}$-trajectories $\sigma_1(t)$ and $\sigma_2(t)$ imply different partition functions with correspondingly shifted Lagrange multipliers. However, as we show (\ref{eq_mean_values_real}) is independent of the particular choice. Indeed, expressing the $\mathcal{T}$-trajectory $\sigma_1(t)$ as $\sigma_1(t)=1/\map Z_1 \exp\left[\ln \sigma_2(t) - \sum_j \omega_j PC_j(t)P \right]$, the asymptotic trajectory (\ref{eq_gibbs form}) can be written as
\begin{eqnarray}
\rho(t) &=& \frac{1}{\map Z}\exp\left(\ln \sigma_1(t) -\sum_j \gamma_j PC_j(t)P\right) \nonumber \\
&=& \frac{1}{\map Z^{'}}\exp\left(\ln \sigma_2(t) - \sum_j \beta_j PC_j(t)P \right) \label{eq_exchange_of_trajectories}
\end{eqnarray}
with $\map Z^{'}=\map Z\map Z_1$ and $\beta_j=\gamma_j+\omega_j$. Therefore, we get
\begin{equation}
\frac{\partial \ln \map Z^{'}}{\partial \beta_j} = \frac{1}{\map Z^{'}}\frac{\partial \map Z^{'}}{\partial \gamma_j}\frac{\partial \gamma_j}{\partial \beta_j} = \frac{1}{\map Z \map Z_1} \map Z_1\frac{\partial \map Z}{\partial \gamma_j}= \frac{\partial \ln \map Z}{\partial \gamma_j}. \nonumber
\end{equation}

Quantum relative entropy of an asymptotic state $\rho(t)$ w.r.t. the given $\map T$-trajectory $\sigma(t)$ is then easily obtained as
\begin{eqnarray}
&& S(\rho(t)|\sigma(t))= \Tr \rho(t) \ln \rho(t) - \Tr \rho(t) \ln \sigma(t) = \nonumber \\
&&-\Tr \rho(t) \ln \map Z + \Tr \left[\rho(t)\left(\ln \sigma(t) - \sum_j \gamma_j PC_j(t)P \right)\right] \nonumber \\ &&- \Tr \rho(t)\ln \sigma(t) = -\ln \map Z - \sum_j \gamma_j \<C_j(t)\>_{\rho(t)}.
\label{eq_quantum_relative_entropy}
\end{eqnarray}
Similarly due to (\ref{eq_mean_values_real}), its change between two infinitesimally closed asymptotic trajectories given by the differential reads
\begin{eqnarray}
&& dS(\rho(t)|\sigma(t))= -\sum_j \left(\frac{\partial \ln \map Z}{\partial \gamma_j}\right) d\gamma_j - \sum_j \<C_j(t)\>_{\rho(t)} d\gamma_j \nonumber \\
&& - \sum_j \gamma_j d\<C_j(t)\>_{\rho(t)} = - \sum_j \gamma_j d\<C_j(t)\>_{\rho(t)}.
\label{eq_change_relative_entropy}
\end{eqnarray}
Note that both quantities (\ref{eq_quantum_relative_entropy}) and (\ref{eq_change_relative_entropy}) are time-independent. Also $S(\rho(t)|\sigma(t+s))$ is time-independent, since $\omega(t)=\sigma(t+s)$ is a $\map T$-trajectory too. We can conclude that in the asymptotic evolution the quantum relative entropy of two evolving states stays constant or is infinite.

Finally, let us explore the quantum relative entropy under the interchange between $\mathcal{T}$-trajectories $\sigma_1(t)$ and $\sigma_2(t)$. Applying (\ref{eq_quantum_relative_entropy}) to (\ref{eq_exchange_of_trajectories}) gives us gradually 
\begin{eqnarray}
&& S(\rho(t)|\sigma_2(t))= -\ln \map Z^{'} - \sum_j \beta_j \<C_j(t)\>_{\rho(t)}= \nonumber \\
&& -\ln \map Z - \ln \map Z_1 -\sum_j (\gamma_j+\omega_j) \<C_j(t)\>_{\rho(t)} = \nonumber \\
&& S(\rho(t)|\sigma_1(t)) + S(\sigma_1(t)|\sigma_2(t)) + \sum_j \omega_j\<C_j(t)\>_{\sigma_1(t)-\rho(t)} \nonumber
\end{eqnarray}
with $\<C_j(t)\>_{\sigma_1(t)-\rho(t)}= \Tr\{C_j(t)(\sigma_1(t)-\rho(t))\}$.
Analogously for the same interchange, the differential of quantum relative entropy reads
\begin{eqnarray}
&& dS(\rho(t)|\sigma_1(t)) + dS(\sigma_1(t)|\sigma_2(t)) \nonumber \\
&& =-\sum_j \gamma_id\<C_j(t)\>_{\rho(t)} - \sum_j \omega_id\<C_j(t)\>_{\sigma_1(t)} \nonumber \\
&& = dS(\rho(t)|\sigma_2(t)) + \sum_j \omega_jd\<C_j(t)\>_{\rho(t)-\sigma_1(t)} . \nonumber
\end{eqnarray}
This provides a useful relation for thermodynamic finite and infinitesimal changes performed in the asymptotic regime.

\section{Conclusion}
\label{sec:conclusion}
In summary, we have studied asymptotics of quantum systems undergoing a quantum Markov (finite-dimensional), discrete or continuous, process (QMP) for different scenarios of available apriori knowledge. We have proven that the asymptotics of QMPs always takes the exponential Gibbs and von Neumann like form and hence can be equivalently formulated as a constrained generalized Jaynes' problem. In order to obtain the correct dynamically induced asymptotic state we must minimize the quantum relative entropy with respect to a properly chosen maximal invariant state or maximal asymptotic trajectory under given expectation values of some constants of motion (including integrals of motion problem. In other words, we show which functional has to be chosen as the proper one when formulating the problem as an extremal one. The results are formulated in three principles, each treating a different situation. The first concerns a known initial state and consequently expectation values of all constants of motion of the given QMP are known. The second deals with a system being in the asymptotic regime for which only some of the expectation values of its constants of motion are known. Finally,  the third case treats the situation where we know, in addition, that the system is in a stationary state and only some  expectation values of some integrals of motion are known. The obtained Jaynes' principles for QMPs coincide with the standard MaxEnt principle if and only if the Markov evolution preserves the maximally mixed state. From this perspective, we have proven the applicability of the MaxEnt principle to equilibrium states of all unital QMPs and generalized it to Jaynes' principle for all asymptotic trajectories of any quantum Markov process. Thus, for QMPs we are able to bridge the gap between the Jaynes' principle as a purely formal information approach to open system dynamics asymptotics and the actual dynamics of the physical process. In particular, Theorem \ref{theorem_jaynes} reveals that asymptotic dynamics of QMPs can be equivalently captured by extremal principle. In addition, Theorem \ref{theorem_incomplete_Jaynes} and \ref{theorem_stationary_incomplete} tackle the issue of incorporating partial knowledge about the asymptotic state in accord with the dynamics of the given QMP.

We found all ingredients needed for the formulation of the extremal principle for all QMPs - the relevant constants of motion, the T-trajectories as well as the quantum relative entropy functional. It results in the equivalence between solving the actual dynamics and the exponential form of the asymptotic (stationary) state solving the asymptotic dynamics. Our approach allows to identify the set of privileged observables needed to define the asymptotic state. The presented method and results are in its spirit very similar to the one applied by Boltzmann to prove the mechanical origin of entropy increase \cite{Boltzmann,Huang,Goldstein,Lenarcic}. The Boltzmann probability distribution and its evolution equation -- the Boltzmann transport equation -- are reincarnated in our formulation to the Markov equation and the collision invariants correspond to the objects we use to define the asymptotic space. As mentioned, the proper variables to be used, are dictated by the system dynamics and can be in many cases found analytically. This is certainly a benefit in particular when methods how to identify the basis of asymptotic dynamics can be given \cite{Kollar2012}. It would be of interest to apply our method for situations of several (interacting) subsystems. In particular to study, what are the essential features of such systems undergoing QMPs, which determine and control the changes of Lagrange multipliers. Much more elusive is the determination of the time scale sufficient to reach the asymptotic regime. The rate in approaching the stationary (equilibrium) regime is strongly dependent on the concrete form of the dynamics and the initial state we start with. Some general constraints for relaxation rates may be found in \cite{Chruscinskiprl,Chruscinskilinalg}, a bound for relaxation rates in large CNOT networks is discussed in \cite{Novotny2021}. This question is particularly interesting for experimentalists, as it defines the natural scale on which the dynamics of the system takes place and when transient effects become negligible.

{\it Acknowledgements:} J. N., J.M. and I. J. have been supported by the
Grant Agency of the Czech Technical University in Prague, grant No. SGS22/181/OHK4/3T/14. 
This publication was funded by the project ”Centre for Advanced Applied Sciences”,
Registry No. CZ.02.1.01/0.0/0.0/16 019/0000778, supported by the Operational Programme Research, Development and Education, co-financed by the European Structural and Investment Funds and the state budget of the Czech Republic. The support for J. N. and I. J. by GAČR of the Czech Republic under Grant No. 23-07169S is gratefully acknowledged.

{\it Data Availability Statement:} No Data associated in the manuscript

\end{document}